\documentclass[a4paper,12pt]{article}
\usepackage[latin1]{inputenc}
\usepackage{amsmath}
\usepackage{amsfonts}
\usepackage{amssymb}
\usepackage{cite}
\usepackage{amsthm}
\usepackage{makeidx}
\usepackage{graphicx}
\usepackage{mathrsfs,xcolor}
\usepackage{enumerate}
\usepackage{blkarray}
\usepackage{bm}
\usepackage{setspace}
\usepackage{algorithm,algorithmic,refcount}
\usepackage[left,pagewise, mathlines]{lineno}
\usepackage{tikz-cd} 
\usepackage{lineno}

\usepackage{tikz}
\usetikzlibrary{arrows,shapes,quotes}
\usetikzlibrary{patterns,decorations.pathreplacing}

\numberwithin{table}{section}
\numberwithin{equation}{section}

\theoremstyle{plain}
\newtheorem{theorem}{Theorem}
\newtheorem{proposition}{Proposition}

\theoremstyle{definition}
\newtheorem{definition}{Definition}
\newtheorem{example}{Example}

\usepackage{authblk}

\author[1,2]{\textbf{Lauro L. Fontanil}}
\author[2,3,4,5]{\textbf{Eduardo R. Mendoza}}
\author[2,*]{\\\textbf{Noel T. Fortun}}

\affil[1]{\small \textit{Institute of Mathematical Sciences and Physics, University of the Philippines,  Los Ba\~{n}os, Laguna 4031, Philippines}}
\affil[2]{ \textit{Mathematics and Statistics Department, De La Salle University, Manila  0922, Philippines}}
\affil[3]{\normalsize\textit{Center for Natural Sciences and Environmental Research, De La Salle University, Manila  0922, Philippines}}
\affil[4]{\normalsize\textit{Max Planck Institute of Biochemistry, Martinsried near Munich, Germany}}
\affil[5]{\normalsize\textit{Faculty of Physics, Ludwig Maximilian University, Munich 80539, Germany}}
\affil[*]{Corresponding author: \texttt{noel.fortun@dlsu.edu.ph}}

\title{\vspace{3.5cm}\textbf{A computational approach to concentration robustness in \\power law kinetic systems of \\Shinar-Feinberg type}}

\date{}

\begin{document}
\maketitle
\begin{abstract}
There have been recent theoretic results that provide sufficient conditions for the existence of a species displaying \textit{absolute concentration robustness} (ACR) in a power law kinetic (PLK) system. One such result involves the detection of ACR among networks of high deficiency by considering a lower deficiency subnetwork with ACR as a local property. In turn, this ``smaller" subnetwork serves as a ``building block" for the larger ACR-possessing network. Here, with this theorem as foundation, we construct an algorithm that systematically checks ACR in a PLK system. By slightly modifying the algorithm, we also provide a procedure that identifies \textit{balanced concentration robustness} (BCR), a weaker form of concentration robustness than ACR, in a PLK system.

\end{abstract}
\baselineskip=0.30in

\setstretch{1.1}

\section{Introduction}

In 2010, Shinar and Feinberg \cite{shin1} introduced the concept of \textit{absolute concentration robustness} (ACR) for a system endowed with mass action kinetics (MAK). A system possesses ACR for a particular species if for every positive steady state of the system, the concentration of that species does not change. Their important contribution is a theorem that identifies some sufficient conditions to guarantee the existence of ACR in a MAK system. 

Almost a decade after, Fortun et al. \cite{fort4} showed that this result can be readily extended to  \textit{power law kinetic systems with reactant-determined interactions} (PL-RDK), a kinetic system more general than MAK. They called this result the Shinar-Feinberg Theorem on ACR  for PL-RDK systems (SFTACR). In a later paper, Fortun and Mendoza \cite{fort3} utilized the concept of dynamic equivalence (through the method discussed in \cite{naza}) and decomposition theory to come up with results motivated by SFTACR. In the same paper, the concept of \textit{balanced concentration robustness} (BCR), a weaker form of concentration robustness than ACR, was also introduced. Furthermore, the study established results that can detect ACR and BCR on some subclasses of PLK system which do not have any deficiency restriction unlike other earlier results on concentration robustness. As future outlook, the paper recommended to build computational approaches that could identify ACR and BCR in a PLK system based on these findings. The said recommendation serves as the motivation of this paper. 

We present here an algorithm that determines if a PLK system with a positive equilibrium and an \textit{independent decomposition} has ACR. We also present an analogous algorithm which can determine if a PLK system having an \textit{incidence independent decomposition} and a complex balanced equilibrium possesses BCR. These algorithms are based on Propositions \ref{3} and \ref{4} (Propositions 8 and 9 in \cite{fort3}). For an ACR or BCR to occur in a PLK-system, each of these propositions require the existence of a building block (which is defined here as a \textit{Shinar-Feinberg} (SF)-\textit{type subnetwork}) satisfying some conditions. We call the algorithms the \textit{building block algorithms for} PLK \textit{systems of} SF-\textit{type}.

These building block algorithms are first of their kind. However, the attempt to analyze a network's concentration robustness computationally is not entirely new. For instance, Kuwahara et al. \cite{kuwa} developed a computer tool which can construct new biochemical networks, endowed with MAK, by combinatorially interchanging user-specified species of existing biochemical networks. The tool was called ACRE which stands for Absolute Concentration Robustness Exploration since it is also capable of recognizing which among the constructed networks have ACR.

The algorithms we have developed employ a \textit{bottom-up} approach which works by initially considering an SF-type subnetwork where the desired decomposition is to be constructed. We chose this approach over a \textit{top-down} approach, which initially searches for an appropriate decomposition with the desired SF-type subnetwork, because the former approach makes the search for a suitable decomposition and subnetwork simpler and computationally cheaper. 



The paper is organized as follows. Fundamental concepts on chemical reaction networks, kinetic systems, and decomposition theory are presented in section 2. Concepts and related results on concentration robustness for PLK sytems are reviewed in section 3. In section 4, the structure of the algorithms as well as the procedures of the ACR building block algorithm are discussed. The BCR building block algorithm and examples are given in section 5. Section 6 provides a summary and an outlook.

\section{Fundamentals of chemical reaction networks and kinetic systems}
\subsection{Structure of chemical reaction networks}

We review in this section some necessary concepts and results on chemical reaction network, the details of which can be found in \cite{fein2, arce, toth, fort3}. 


\begin{definition} A \textbf{chemical reaction network} is a triple $\mathscr N = (\mathscr S, \mathscr C, \mathscr R)$ of three non-empty finite sets:
\begin{enumerate}
    \item a set \textbf{species} $\mathscr S$;
    \item a set $\mathscr C$ of \textbf{complexes}, which are nonnegative integer linear combinations of the species; and
    \item a set $\mathscr R \subseteq \mathscr C \times \mathscr C$ of reactions such that
    \begin{itemize}
        \item $(y,y)\notin \mathscr R$ for all $y\in \mathscr C$, and
        \item for each $y\in \mathscr C$, there exists a $y'\in \mathscr C$ such that $(y,y')\in \mathscr R$ or $(y',y)\in \mathscr R$.
    \end{itemize}
\end{enumerate}
\end{definition}

\noindent The nonnegative coefficients of the species in a complex are referred to as \textbf{stoichiometric coefficients}.  In this paper, a reaction $R_i=(y_i,y'_i)$ is also denoted by $R_i: y_i\rightarrow y'_i$ and $y_i$ and $y'_i$ are called the \textbf{reactant} and \textbf{product complexes} of $R_i$, respectively. Further, we reserve the letters $m, n$, and $r$ to denote the numbers of species, complexes, and reactions, respectively. The following example shows that a CRN can be represented by a digraph where the complexes and reactions serve as the digraph's vertices and arcs, respectively. 

\subsection*{Running example 1}

Consider the CRN represented by the digraph below. 

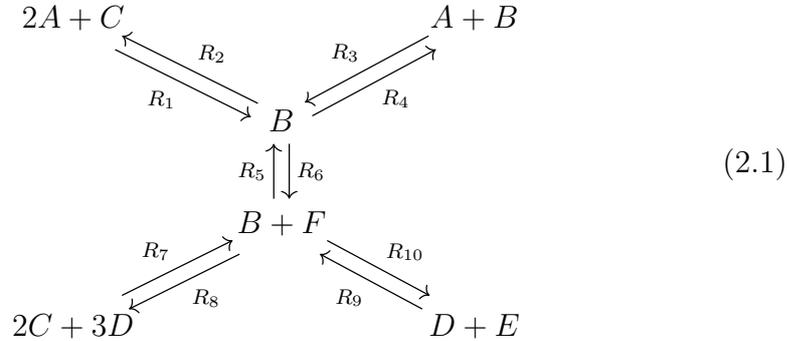
\begin{equation}
\begin{tikzcd}
2A+C \arrow[rd, "R_1"', shift right] &                                                                                                 & A+B \arrow[ld, "R_3"', shift right] \\
                                     & B \arrow[lu, "R_2"', shift right] \arrow[ru, "R_4"', shift right] \arrow[d, "R_6", shift left]  &                                     \\
                                     & B+F \arrow[u, "R_5", shift left] \arrow[ld, "R_8", shift left] \arrow[rd, "R_{10}", shift left] &                                     \\
2C+3D \arrow[ru, "R_7", shift left]  &                                                                                                 & D+E \arrow[lu, "R_9", shift left]  
\end{tikzcd}
\label{graph100}
\end{equation}

\noindent This CRN has $m=6, n=6$, and $r=10$. It also has the following sets:

\begin{align*} 
    \mathscr{S} &= \{ A,B,C,D,E,F \}; \\ 
    \mathscr{C} &= \{2A+C, A+B, B, B+F, 2C+3D, D+E\}; {\textnormal {and}}\\
    \mathscr{R} &= \{R_1, R_2, R_3, R_4, R_5, R_6, R_7, R_8, R_9, R_{10} \}.
\end{align*} 

We call the connected components of a CRN \textbf{linkage classes} whose number is denoted by $\ell$. The \textbf{strong linkage classes} refer to the strongly connected components and the \textbf{terminal strong linkage classes} are the strongly connected components without outgoing arrows. If each of the linkage classes of a CRN is also a strong linkage class, we call it \textbf{weakly reversible}. The CRN in Running example 1 is weakly reversible since it has only one linkage class which is also a strong linkage class. A complex is called \textbf{terminal} if it belongs to a terminal strong linkage class.

CRNs are studied with the aid of finite dimensional spaces $\mathbb R^{\mathscr S}$, $\mathbb R^{\mathscr C}$, and $\mathbb R^{\mathscr R}$, respectively referred to as \textbf{species space, complex space}, and \textbf{reaction space}. For every reaction $y\rightarrow y'$, we associate a vector, called \textbf{reaction vector}, resulted from subtracting the reactant complex $y$ from the product complex $y'$. The linear subspace of $\mathbb R^{\mathscr S}$ defined by $S:=\text{span}\{y'-y\in \mathbb R^{\mathscr S}|y\rightarrow y'\in \mathscr R\}$ is called the \textbf{stoichiometric subspace} $S$ of a CRN. Its dimension $s$ also refers to the \textbf{rank }of the CRN. A CRN can be characterized by a nonnegative integer $\delta$ called \textbf{deficiency} which is given by $\delta=n-\ell-s$. The CRN in Running example 1 has $S = \text{span}\{r_1, r_2, r_3, r_4, r_5\}$ where

\begin{center}
$r_1 = \begin{bmatrix} 2\\-1\\1\\0\\0\\0\\ \end{bmatrix}, r_2 = \begin{bmatrix} 1\\0\\0\\0\\0\\0\\ \end{bmatrix}, r_3 = \begin{bmatrix} 0\\0\\0\\0\\0\\1\\ \end{bmatrix}, r_4 = \begin{bmatrix} 0\\-1\\2\\3\\0\\-1\\ \end{bmatrix}, r_5 = \begin{bmatrix} 0\\-1\\0\\1\\1\\-1\\ \end{bmatrix}.$
\end{center}

\noindent Furthermore, its rank and deficiency is $s=5$ and $\delta=6-1-5=0$, respectively.


\begin{definition} Let $\mathscr N = (\mathscr S, \mathscr C, \mathscr R)$ be a CRN. The \textbf{incidence map} $I_a: \mathbb R^{\mathscr R}\rightarrow \mathbb R^{\mathscr C}$ is the linear map defined by mapping for each reaction $R_i: y_i\rightarrow y'_i\in \mathscr R$, the basis vector $\omega_i$ to the vector $\omega_{y'_i}-\omega_{y_i}\in \mathscr C$.
\end{definition}

\noindent As a linear map, the incidence map has an $n\times r$ matrix representation, called \textbf{incidence matrix}, whose entries are described by the following.

\begin{center}
    $(I_a)_{(i,j)}= \left\{ \begin{array}{rl}
    -1 & \mbox{if $i$ is the reactant complex of reaction $j \in \mathscr R,$} \\ 1 & \mbox{if $i$ is the product complex of reaction $j \in \mathscr R,$} \\ 0 & \mbox{otherwise.} \\
\end{array}\right.$
\end{center}

The CRN in Running example 1 has the following incidence matrix.

\begin{equation}
I_a = 
    \begin{blockarray}{*{10}{c} l}
        \begin{block}{*{10}{>{$\footnotesize}c<{$}} l}
            $R_1$ & $R_2$ & $R_3$ & $R_4$ & $R_5$ & $R_6$ & $R_7$  & $R_8$ & $R_9$ & $R_{10}$ \\
        \end{block}
        \begin{block}{[*{10}{c}]>{$\footnotesize}l<{$}}
            -1& 1 & 0 & 0 & 0 & 0 & 0 & 0 & 0 & 0 &  $2A+C$\\
            0 & 0& -1& 1 & 0 & 0 & 0 & 0 & 0 & 0 &  $A+B$\\
            1 & -1 & 1 & -1& 1 & -1& 0 & 0 & 0 & 0 &  $B$\\
            0 & 0 & 0 & 0 & -1& 1 & 1 & -1& 1 & -1&  $B+F$\\
            0 & 0 & 0 & 0 & 0 & 0 & -1& 1 & 0 & 0 &  $2C+3D$\\
            0 & 0 & 0 & 0 & 0 & 0 & 0 & 0 & -1& 1 &  $D+E$\\
        \end{block}
    \end{blockarray}
\label{eqn9}
\end{equation}

\noindent For the incidence matrix (of a directed graph), it is known that $\text{dim Im} I_a = n -\ell$ \cite{boros2013}.

\subsection{Dynamics of chemical reaction networks}

A \textbf{kinetics} is an assignment of a rate function to each reaction in a CRN. A network $\mathscr N$ together with a kinetics $K$ is called a \textbf{chemical kinetic system} (CKS) and is denoted here by $(\mathscr N,K)$. 
\textbf{Power law kinetics} (PLK) is identified by the \textbf{kinetic order matrix} which is an $r\times m$ matrix $F=[F_{ij}]$, and vector $k\in \mathbb R^{\mathscr R}_{>0}$, called the \textbf{rate vector}.

\begin{definition} A kinetics $K: \mathbb R^{\mathscr S}_{>0} \rightarrow \mathbb R^{\mathscr R}$ is a \textbf{power law kinetics} if

\begin{center}
    $K_i(x)=k_ix^{F_{i,\cdot}}$ for $i=1,\dots, r$
\end{center}

\noindent where $k_i\in \mathbb R_{>0}$, $F_{i,j} \in \mathbb R$, and $F_{i,\cdot}$ is the row of $F$ associated to reaction $R_i$.
\end{definition}

We can classify a PLK system based on the kinetic orders assigned to its \textbf{branching reactions} (i.e., reactions sharing a common reactant complex). 

\begin{definition}
A PLK system has \textbf{reactant-determined kinetics} (of type PL-RDK) if for any two branching reactions $R_i, R_j\in \mathscr R$, the corresponding rows of kinetic orders in $F$ are identical, i.e., $F_{ih}=F_{jh}$ for $h=1, \dots,m$. Otherwise, a PLK system has \textbf{non-reactant-determined kinetics} (of type PL-NDK).
\end{definition}

Suppose the CRN in Running example 1 is a PLK system with the following kinetic order matrix.

\begin{equation}
F =
  \begin{blockarray}{*{6}{c} l}
    \begin{block}{*{6}{>{$\footnotesize}c<{$}} l}
      $A$ & $B$ & $C$ & $D$ & $E$ & $F$ \\
    \end{block}
    \begin{block}{[*{6}{c}]>{$\footnotesize}l<{$}}
      2 & 0 & 1 & 0 & 0 & 0 & $R_1$ \\
      0 & 2 & 0 & 0 & 0 & 0 & $R_2$ \\
      1 &0.5& 0 & 0 & 0 & 0 & $R_3$ \\
      0 & 1 & 0 & 0 & 0 & 0 & $R_4$ \\
      0 & 1 & 0 & 0 & 0 & 1 & $R_5$ \\
      0 & 1 & 0 & 0 & 0 & 0 & $R_6$ \\
      0 & 0 & 2 & 3 & 0 & 0 & $R_7$ \\
      0 &0.5& 0 & 0 & 0 & 1 & $R_8$ \\
      0 & 0 & 0 & 1 & 1 & 0 & $R_9$ \\
      0 & 1 & 0 & 0 & 0 & 1 & $R_{10}$ \\
    \end{block}
  \end{blockarray}
  \label{graph300} 
\end{equation}

\noindent $R_2$ and $R_4$ are two branching reactions whose corresponding rows in $F$ (or \textbf{kinetic order vectors}) are not the same since $F_{22}=2\neq1=F_{42}$. Hence, the system is of type PL-NDK.

The well-known \textbf{mass action kinetic system} (MAK) forms a subset of PL-RDK systems. In particular, MAK is given by $K_i(x)=k_ix^{Y_{.,j}}$ for all reactions $R_i: y_i \rightarrow y'_i \in \mathscr R$ with $k_i\in \mathbb R_{>0}$ (called \textbf{rate constant}). The vector $Y_{.,j}$ contains the stoichiometric coefficients of a reactant complex $y_i\in \mathscr C$.


\begin{definition} The \textbf{species formation rate function} of a chemical kinetic system is the vector field

\begin{center}
    $f(c)=NK(c)=\displaystyle \sum_{y_i\rightarrow y'_i\in \mathscr R}K_i(c)(y'_i-y_i)$, where $c\in \mathbb R^{\mathscr S}_{\geq 0}$,
\end{center}
where $N$ is the $m\times r$ matrix, called \textbf{stoichiometric matrix},  whose columns are the reaction vectors of the system. \noindent The equation $dc/dt=f(c(t))$ is the \textbf{ODE or dynamical system} of the chemical kinetic system. An element $c^*$ of $\mathbb R^{\mathscr S}_{>0}$ such that $f(c^*)=0$ is called a \textbf{positive equilibrium} or \textbf{steady state} of the system. We use $E_+(\mathscr N,K)$ to denote the set of all positive equilibria of a CKS.
\end{definition}

Analogous to the species formation rate function, we also have the complex formation rate function.

\begin{definition} The \textbf{complex formation rate function} $g: \mathbb R^{\mathscr S}_{>0}\rightarrow \mathbb R^{\mathscr C}$ of a chemical kinetic system is the vector field

\begin{center}
    $g(c)=I_aK(c)=\displaystyle \sum_{y_i\rightarrow y'_i\in \mathscr R}K_i(c)(\omega_{y'_i}-\omega_{y_i})$, where $c\in \mathbb R^{\mathscr S}_{\geq 0}$.
\end{center}

\noindent where $I_a$ is the incidence map. A CKS is \textbf{complex balanced} if it has complex balanced steady state, i.e., there is a composition $c^{**}\in \mathbb R_{>0}^{\mathscr S}$ such that $g(c^{**} )=0$. We denote by $Z_+(\mathscr N,K)$ the set of all complex balanced steady states of the system.
\end{definition}



\begin{theorem}[Corollary 4.8, \cite{fein3}] \label{1.5}  If a CKS has deficiency 0, then its steady states are all complex balanced.
\end{theorem}

\subsection{Decomposition theory}

We briefly discuss here two decomposition types which are important components of our findings.

\begin{definition} Let $\mathscr N = (\mathscr S, \mathscr C, \mathscr R)$ be a CRN. A \textbf{covering} of $\mathscr N$ is a collection of subsets \{$\mathscr R_1, \dots, \mathscr R_k$\} whose union is $\mathscr R$. A covering is called a \textbf{decomposition} of $\mathscr N$ if the sets $\mathscr R_i$ form a partition of $\mathscr R$.
\end{definition}

 $\mathscr R_i$ defines a subnetwork $\mathscr N_i$ of $\mathscr N$ where $\mathscr N_i=(\mathscr S_i, \mathscr C_i, \mathscr R_i)$ such that $\mathscr C_i$ consists of all complexes occurring in $\mathscr R_i$ and $\mathscr S_i$ has all the species occurring in $\mathscr C_i$. From here on, we will use $\{\mathscr N_1, \mathscr N_2, \dots, \mathscr N_k\}$ to denote the decomposition induced by $\mathscr R_i$'s. Also, for convenience, we will sometimes refer to $\mathscr N$, given that it has been ``decomposed'', as the parent network.

The idea of network decomposition originally came from Feinberg \cite{fein1}. He introduced a decomposition subclass called independent decomposition, which is identified based on the subnetworks' stoichiometric subspaces. 

\begin{definition}
A decomposition is \textbf{independent} if $S$ is the direct sum of the subnetworks' stoichiometric subspaces $S_i$ or equivalently, $s=s_1+\cdots+s_k$.
\end{definition}

Consider the decomposition $\{\mathscr N_1, \mathscr N_2, \mathscr N_3\}$ of the CRN in Running example 1 (see (\ref{graph200})). Notice that $\mathscr N_1, \mathscr N_2$, and $\mathscr N_3$ have $1, 1$, and $3$ as ranks, respectively. Since the parent network is of rank 5, the decomposition is independent.

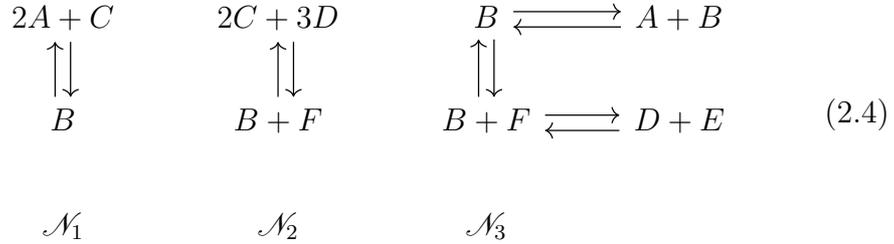
\begin{equation}
\begin{tikzcd}
2A+C \arrow[d, shift left] & 2C+3D \arrow[d, shift left=2] & B \arrow[d, shift left] \arrow[r, shift left]   & A+B \arrow[l, shift left] \\
B \arrow[u, shift left]    & B+F \arrow[u]                 & B+F \arrow[u, shift left] \arrow[r, shift left] & D+E \arrow[l, shift left] \\
\mathscr N_1               & \mathscr N_2                  & \mathscr N_3                                    &                  \end{tikzcd}
\label{graph200}
\end{equation}

The following observation relates the images of the incidence maps of the parent network and its subnetworks.

\begin{proposition}[Proposition 6, \cite{fari}] \label{2}  Let $\{\mathscr N_1, \mathscr N_2, \cdots, \mathscr N_k\}$ be a covering of $\mathscr N$. Further, let $I_{a,i}$ be the incidence map of the subnetwork $\mathscr N_i$ with $n_i-\ell_i= \text{dim  }I_{a,i}$. Then,
\begin{enumerate} [i.]
    \item $\text{Im } I_a=\text{Im } I_{a,1}+\text{Im } I_{a,2}+\cdots+\text{Im } I_{a,k}$
    \item $n-\ell\leq (n_1-\ell_1) +(n_2-\ell_2)+\cdots+(n_k-\ell_k) $
\end{enumerate}
\end{proposition}

Motivated by the idea of independent decomposition and Proposition \ref{2}, Farinas et al. \cite{fari} introduced the concept of incidence independent decomposition.

\begin{definition} A decomposition of a CRN $\mathscr N$ is \textbf{incidence independent} if and only if the image of the incidence map of $\mathscr N$ is the direct sum of the images of the incidence maps of the subnetworks or equivalently, $n-\ell = \sum (n_i-\ell_i)$.
\end{definition}

The decomposition of the CRN in Running example 1 whose subnetworks are given in (\ref{graph200}) is incidence independent since $n-\ell=5$, $n_1 -\ell_1=n_2-\ell_2=1$, and $n_3-\ell_3=3$.

Interestingly, incidence independent decompositions are more commonly observed in CRNs compared to independent decompositions. In particular, linkage class decomposition (i.e., a decomposition where the subnetworks are the linkage classes) is always incidence independent but not necessarily independent \cite{fari}.

\section{Concentration robustness in PLK systems of Shinar-Feinberg type}

In this section, we first review concepts and results on concentration robustness in PLK systems. We then introduce a running example to illustrate the different steps of the algorithm in the following sections. 

\subsection{Review of concentration robustness in PLK systems of Shinar-Feinberg type}

Details of the concepts and results reviewed below can be found in Fortun and Mendoza \cite{fort3}.

\begin{definition}
A PLK system $(\mathscr N, K)$ has \textbf{absolute concentration robustness} in a species $X\in \mathscr S$ if there exists $c^*\in E_+(\mathscr N, K)$ and for every other $c^{**}\in  E_+(\mathscr N, K)$, we have $c^*_X=c^{**}_X$.
\end{definition}

Balanced concentration robustness (BCR), a concentration robustness that is weaker than ACR, was also introduced in \cite{fort3}.  

\begin{definition}
A complex balanced CKS $(\mathscr N, K)$ has \textbf{balanced concentration robustness} in a species $X \in \mathscr S$ if $X$ has the same value for all $c\in Z_+(\mathscr N, K)$.
\end{definition}

If a system has ACR in a species, then it also has BCR for that species. The converse, however, of this observation is not always true.


\begin{definition}
A pair of reactions in a PLK system is called a \textbf{Shinar-Feinberg pair} (or SF-pair) in a species $X$ if their kinetic order vectors differ only in $X$. A network that contains an SF-pair is called a \textbf{Shinar-Feinberg type} (SF-type).
\end{definition}


A reactant complex in a PL-NDK system of some branching reactions is said to be an \textbf{NDK node} if the kinetic order vectors of the branching reactions are different. A \textbf{complex factorizable} (CF)-\textbf{subset} of a reactant complex is a subset of the branching reactions whose kinetic order vectors are the same. 

\begin{definition}
A PL-NDK system containing a single NDK node which has two CF-subsets, at least one of which contains only one reaction, is called \textbf{minimally PL-NDK}.
\end{definition}

The next proposition enables us to determine if a PLK system has ACR. Unlike other similar results on ACR \cite{shin1, fort3, fort4}, this proposition does not have any deficiency restriction imposed on the underlying network.  This allows us to deal with higher deficiency systems including higher deficiency MAK systems. 

\begin{proposition}[Proposition 8, \cite{fort3}] \label{3}  Let $(\mathscr N, K)$ be a {\normalfont{PLK}} system with a positive equilibrium and an independent decomposition $\{\mathscr N_1, \mathscr N_2, \cdots, \mathscr N_k\}$. If there is an $\mathscr N_i$ with $(\mathscr N_i, K_i)$ of {\normalfont{SF-type}} in $X\in \mathscr S$ such that 

\begin{enumerate} [i.]
    \item $\delta_i=0$ and is {\normalfont{PL-RDK}} or minimally {\normalfont{PL-NDK}}, or
    \item $\delta_i=1$ and is {\normalfont{PL-RDK}}.
\end{enumerate}

\noindent Then $(\mathscr N, K)$ has {\normalfont{ACR}} in $X$.
\end{proposition}

The next proposition, which is analogous to Proposition \ref{3}, can determine the existence of BCR. 

\begin{proposition}[Proposition 9, \cite{fort3}] \label{4}  Let $(\mathscr N, K)$ be a {\normalfont{PLK}} system with a complex balanced equilibrium and an incidence independent decomposition $\{\mathscr N_1, \mathscr N_2, \cdots, \mathscr N_k\}$. If there is an $\mathscr N_i$ with $(\mathscr N_i, K_i)$ of {\normalfont{SF-type}} in $X\in \mathscr S$ such that 

\begin{enumerate} [i.]
    \item $\delta_i=0$ and is {\normalfont{PL-RDK}} or minimally {\normalfont{PL-NDK}}, or
    \item $\delta_i=1$ and is {\normalfont{PL-RDK}}.
\end{enumerate}

\noindent Then $(\mathscr N, K)$ has {\normalfont{BCR}} in $X$.
\end{proposition}

For a PLK system to have the desired concentration robustness, Proposition \ref{3} requires the system to have a positive equilibrium and independent decomposition while Proposition \ref{4} requires it to have an incidence independent decomposition and a complex balanced equilibrium. In addition, both requires having a \textbf{building block} or an SF-type subnetwork, taken from the attained decomposition, which satisfies certain conditions. Notice that after ensuring that a system has the desired decomposition and equilibrium, the said propositions translate the problem of identifying ACR or BCR to a search for a building block, eventually shifting the focus of the analysis on this smaller subnetwork. 

\subsection*{Running example 2}

We consider the following MAK system representing the early STAT signaling network coupled with the receptor complex formation upon interferon (IFN) induction which was studied in \cite{oter} (see (\ref{eqn3001}) for the CRN and (\ref{graph3002}) for the kinetic order matrix). It has a positive equilibrium, but its deficiency is 3.

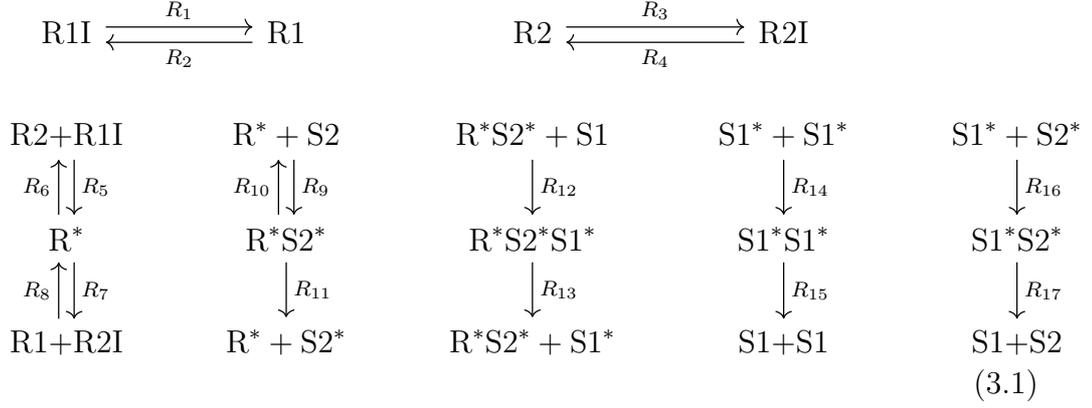
\begin{equation}
\begin{tikzcd}
\textrm{R1I} \arrow[r, "R_1", shift left]                              & \textrm{R1} \arrow[l, "R_2", shift left]                                      & \textrm{R2} \arrow[r, "R_3", shift left]                  & \textrm{R2I} \arrow[l, "R_4", shift left]       &                                                 \\
\textrm{R2+R1I} \arrow[d, "R_5", shift left]                           & \textrm{R}^*+\textrm{S2} \arrow[d, "R_9", shift left]                         & \textrm{R}^*\textrm{S2}^*+\textrm{S1} \arrow[d, "R_{12}"] & \textrm{S1}^*+\textrm{S1}^* \arrow[d, "R_{14}"] & \textrm{S1}^*+\textrm{S2}^* \arrow[d, "R_{16}"] \\
\textrm{R}^* \arrow[u, "R_6", shift left] \arrow[d, "R_7", shift left] & \textrm{R}^*\textrm{S2}^* \arrow[u, "R_{10}", shift left] \arrow[d, "R_{11}"] & \textrm{R}^*\textrm{S2}^*\textrm{S1}^* \arrow[d, "R_{13}"]  & \textrm{S1}^*\textrm{S1}^* \arrow[d, "R_{15}"]  & \textrm{S1}^*\textrm{S2}^* \arrow[d, "R_{17}"]  \\
\textrm{R1+R2I} \arrow[u, "R_8", shift left]                           & \textrm{R}^*+\textrm{S2}^*                                                    & \textrm{R}^*\textrm{S2}^*+\textrm{S1}^*                   & \textrm{S1+S1}                                  & \textrm{S1+S2}                                 
\end{tikzcd}
\label{eqn3001}
\end{equation}

\begin{equation}
F =
\scalebox{.85}{
  \begin{blockarray}{*{13}{c} l}
    \begin{block}{*{13}{>{$\footnotesize}c<{$}} l}
      \textrm{R1} & \textrm{R1I} & \textrm{R2} & \textrm{R2I} & \textrm{R}^* & \textrm{S1} & \textrm{S1}^* & \textrm{S2} & \textrm{S2}^* & \textrm{S1}^*\textrm{S1}^* & \textrm{S1}^*\textrm{S2}^* & \textrm{R}^*\textrm{S2}^* & \textrm{R}^*\textrm{S2}^*\textrm{S1}^* \\
    \end{block}
    \begin{block}{[*{13}{c}]>{$\footnotesize}l<{$}}
      0 & 1 & 0 & 0 & 0 & 0 & 0 & 0 & 0 & 0 & 0 & 0 & 0 & R_1 \\
      1 & 0 & 0 & 0 & 0 & 0 & 0 & 0 & 0 & 0 & 0 & 0 & 0 & R_2 \\
      0 & 0 & 1 & 0 & 0 & 0 & 0 & 0 & 0 & 0 & 0 & 0 & 0 & R_3 \\
      0 & 0 & 0 & 1 & 0 & 0 & 0 & 0 & 0 & 0 & 0 & 0 & 0 & R_4 \\
      0 & 1 & 1 & 0 & 0 & 0 & 0 & 0 & 0 & 0 & 0 & 0 & 0 & R_5 \\
      0 & 0 & 0 & 0 & 1 & 0 & 0 & 0 & 0 & 0 & 0 & 0 & 0 & R_6 \\
      0 & 0 & 0 & 0 & 1 & 0 & 0 & 0 & 0 & 0 & 0 & 0 & 0 & R_7 \\
      1 & 0 & 0 & 1 & 0 & 0 & 0 & 0 & 0 & 0 & 0 & 0 & 0 & R_8 \\
      0 & 0 & 0 & 0 & 1 & 0 & 0 & 1 & 0 & 0 & 0 & 0 & 0 & R_9 \\
      0 & 0 & 0 & 0 & 0 & 0 & 0 & 0 & 0 & 0 & 0 & 1 & 0 & R_{10} \\
      0 & 0 & 0 & 0 & 0 & 0 & 0 & 0 & 0 & 0 & 0 & 1 & 0 & R_{11} \\
      0 & 0 & 0 & 0 & 0 & 1 & 0 & 0 & 0 & 0 & 0 & 1 & 0 & R_{12} \\
      0 & 0 & 0 & 0 & 0 & 0 & 0 & 0 & 0 & 0 & 0 & 0 & 1 & R_{13} \\
      0 & 0 & 0 & 0 & 0 & 0 & 2 & 0 & 0 & 0 & 0 & 0 & 0 & R_{14} \\
      0 & 0 & 0 & 0 & 0 & 0 & 0 & 0 & 0 & 1 & 0 & 0 & 0 & R_{15} \\ 
      0 & 0 & 0 & 0 & 0 & 0 & 1 & 0 & 1 & 0 & 0 & 0 & 0 & R_{16} \\
      0 & 0 & 0 & 0 & 0 & 0 & 0 & 0 & 0 & 0 & 1 & 0 & 0 & R_{17} \\
    \end{block}
  \end{blockarray}}
  \label{graph3002} 
\end{equation}

Take the decomposition $\{\mathscr N_1, \mathscr N_2\}$ of the system where $\mathscr R_1=\{R_1, R_3, R_5\}$ and $\mathscr R_2=\{R_9, R_{11}, R_{12}, R_{13}, R_{14}, R_{16}\}$ (see (\ref{eqn2})). The decomposition is independent since $\mathscr N_1$ and $\mathscr N_2$ have, respectively, $3$ and $6$ as ranks while the given system is of rank $9$. Observe that $\mathscr N_1$ is an SF-type subnetwork for containing reactions $R_1$ and $R_{5}$ which are SF-pair in R2 since their kinetic order vectors differ only in that species. The subsystem $(\mathscr N_1, K_1)$ is of deficiency $\delta_1 = 7 - 3 -3 = 1$ and is clearly a PL-RDK implying that it has ACR in $R2$. 

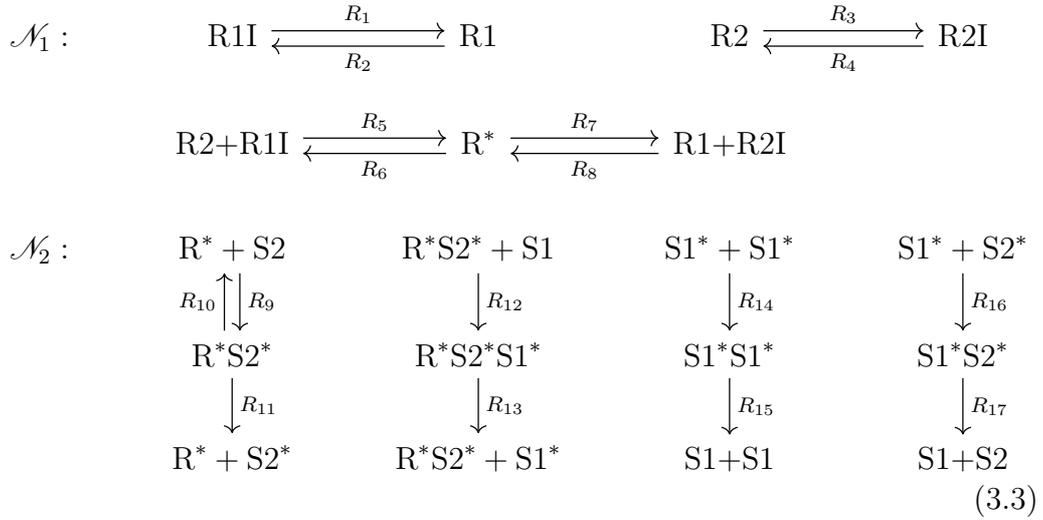
\begin{equation}
\begin{tikzcd}
\mathscr N_1: & \textrm{R1I} \arrow[r, "R_1", shift left]                                     & \textrm{R1} \arrow[l, "R_2", shift left]                               & \textrm{R2} \arrow[r, "R_3", shift left]        & \textrm{R2I} \arrow[l, "R_4", shift left]       \\
              & \textrm{R2+R1I} \arrow[r, "R_5", shift left]                                  & \textrm{R}^* \arrow[l, "R_6", shift left] \arrow[r, "R_7", shift left] & \textrm{R1+R2I} \arrow[l, "R_8", shift left]    &                                                 \\
\mathscr N_2: & \textrm{R}^*+\textrm{S2} \arrow[d, "R_9", shift left]                         & \textrm{R}^*\textrm{S2}^*+\textrm{S1} \arrow[d, "R_{12}"]              & \textrm{S1}^*+\textrm{S1}^* \arrow[d, "R_{14}"] & \textrm{S1}^*+\textrm{S2}^* \arrow[d, "R_{16}"] \\
              & \textrm{R}^*\textrm{S2}^* \arrow[u, "R_{10}", shift left] \arrow[d, "R_{11}"] & \textrm{R}^*\textrm{S2}^*\textrm{S1}^* \arrow[d, "R_{13}"]               & \textrm{S1}^*\textrm{S1}^* \arrow[d, "R_{15}"]  & \textrm{S1}^*\textrm{S2}^* \arrow[d, "R_{17}"]  \\
              & \textrm{R}^*+\textrm{S2}^*                                                    & \textrm{R}^*\textrm{S2}^*+\textrm{S1}^*                                & \textrm{S1+S1}                                  & \textrm{S1+S2}                                 
\end{tikzcd}
\label{eqn2}
\end{equation}

\section{The ACR building block algorithm for PLK systems of Shinar-Feinberg type}

\subsection{General structure and flow of the algorithm}

Propositions \ref{3} and \ref{4} require four things in a PLK system to guarantee the existence of ACR or BCR. These are: (1) the system must have an equilibrium (positive/complex balanced); (2) the system must have a decomposition (independent/incidence independent); (3) the decomposition must have an SF-type subnetwork; and (4) the subnetwork must be a building block. These requirements are the key factors that we have considered in the development of the algorithms. The major stages of the algorithms dwell on the checking of these requirements (see Table \ref{stages}). 

\begin{table}[h]
\centering
\begin{tabular}{|l|l|lll}
\cline{1-2}
\multicolumn{1}{|r|}{\textbf{STAGE 1}} & \begin{tabular}[c]{@{}l@{}}Check if the PLK system has multiple\\ equilibria \end{tabular} &  &  &  \\  \cline{1-2}

\textbf{STAGE 2}                       & \begin{tabular}[c]{@{}l@{}}Check if the PLK system has an SF-type\\ subnetwork.\end{tabular}                                     &  &  &  \\ \cline{1-2}

\textbf{STAGE 3}                       & \begin{tabular}[c]{@{}l@{}}Check if the PLK system has an appropriate \\ decomposition induced by the SF-type subnetwork.\end{tabular}                                &  &  &  \\ \cline{1-2}

\textbf{STAGE 4}                       & \begin{tabular}[c]{@{}l@{}}Check if the SF-type subnetwork is a building \\  block.\end{tabular}                                    &  &  &  \\ \cline{1-2}
\end{tabular}
\caption{Stages of the algorithm.}
\label{stages}
\end{table}

Since the first stage includes equilibrium checking, we clarify that this paper does not focus on the checking of the capacity of a system to admit an equilibrium. Readers are advised to refer to papers such as \cite{fort1, fort2, hern} for discussions regarding this topic. Furthermore, in  this stage, we require the system to have an equilibrium that is not unique in the whole species space, otherwise the system will trivially have a concentration robustness in all species \cite{fort3}. 

With these stages in mind, we developed the algorithms employing a \textbf{bottom-up} (BU) approach. This approach works by initially identifying a subnetwork which contains an SF-pair beginning with the smallest possible such subnetwork, i.e., the subnetwork consisting only the SF-pairs. Finding every SF-pair is done by tracking the difference of every two kinetic order vectors in the system. A decomposition is constructed from the identified SF-type subnetwork. Once the desired decomposition has been obtained, the subnetwork is subjected to a series of tests (called \textbf{building block tests}) to determine if the necessary conditions enumerated in Propositions \ref{3} and \ref{4} are met. In case the chosen subnetwork does not give the desired output, another subnetwork is considered and tested until all the possible subnetworks containing the SF-pairs have been checked. 

Finding a suitable decomposition from of an SF-type subnetwork is not always an easy task. The challenge is even greater when the network system at hand is large. Nevertheless, a good algorithm needs to have an efficient mechanism in finding such decomposition. Apparently, it is not efficient to consider a decomposition with arbitrary number of subnetworks. Hence, we must consider a decomposition with the smallest possible number of subnetworks that will do the job. This led us to design the algorithms to just consider binary decompositions. This decision is further justified by the following reasons. First, a more refined decomposition is difficult to handle especially if we want to achieve efficiency. Second, note that we are only interested in the subnetwork which contains the SF-pair leaving the other subnetworks with no purpose in the analysis. So, having a decomposition with many subnetworks leads to some impractical labor loss. For the last reason, we turn our attention to the following results of Farinas et al. \cite{fari} on the coarsening of independent and incidence independent decompositions. The results are stated as follows. 

\begin{proposition}[Propositions 5 and 8, \cite{fari}] \label{5} If a decomposition is independent $($incidence independent$)$, then any coarsening of the decomposition is also independent $($incidence independent$)$. 
\end{proposition}

These results mean that a network possessing an independent/incidence independent decomposition also possesses a binary independent/incidence independent decomposition. In particular, given that a network $\mathscr N$ has an independent/incidence independent decomposition say $\{\mathscr N_1, \mathscr N_2, \cdots, \mathscr N_k \}$ where  $\mathscr N_1$ contains an SF-pair. Then, $\{\mathscr N_1, \mathscr N^* \}$, where $\mathscr N^*=\mathscr N_2 \cup \mathscr N_3 \cup \cdots \cup \mathscr N_k$, must also be an independent/incidence independent decomposition. These tell us that it does not matter if we consider a binary decomposition since the independence or incidence independence of a finer decomposition is inherited by it. 

It is also important to note that even if we can decompose a network in so many ways there is no guarantee that it possesses an independent/incidence independent decomposition. Take  network $\mathscr N = (\mathscr S, \mathscr C, \mathscr R)$ with $\mathscr S = \mathscr C = \{A, B\}$ and $\mathscr R = \{ A \rightarrow B, B\rightarrow A\}$ as an example. Observe that $\mathscr N$ can be only decomposed as $\{\mathscr N_1, \mathscr N_2\}$ where $\mathscr R_1 = \{A\rightarrow B\}$ and $\mathscr R_2 = \{B\rightarrow A\}$. However, this decomposition is not independent nor incidence independent. So, Propositions \ref{3} and \ref{4} can no longer handle this network. Consequently, this kind of CRN limits the scope of the algorithms. This implies that the algorithms may determine if a system has the desired concentration robustness but not the nonexistence of such.

On the other hand, a network may have so many independent or incidence independent decompositions. Unfortunately, it is not always the case that one decomposition gives everything we need to have an ACR or BCR. This means that a good algorithm must be capable of analyzing every possible independent/incidence decomposition a network may have. In appendix, an example which illustrates how one of our algorithms failed to show if the network has the desired concentration robustness is given. However, it also shows how it is able to generate and analyze every necessary binary decomposition of the network.

\subsection{Detailed step-by-step description of the ACR building block algorithm}
The following are the details on how we proceed with the ACR building block algorithm for PLK systems of Shinar-Feinberg type.

\begin{enumerate}
    \item 	The inputs are the network $\mathscr N$ and its kinetic order matrix $F$.
    \item 	From the inputs, vital information including the set of reaction vectors $\mathscr R^* = \{y'-y \in \mathbb R^{\mathscr S}|y\rightarrow y' \in \mathscr R\}$ and the network's rank will be collected (STEP 1).
    \item Every SF-pair will also be collected and stored (STEP 2).
    \item 	From an SF-pair $\{R_u, R_v\}$, corresponding vectors $r_u$ and $r_v$ will be obtained. Then, $\{r_u,r_v\}$ will be extended to a basis $B = \{r_u,r_v,r'_1,r'_2,\cdots, r'_p\}$ of the stoichiometric subspace $S$. The rank obtained in STEP 1 determines the number of elements of $B$ (STEPS 3-5).
    \item 	After setting $B_1=\{r_u,r_v\}$ and $B_2=\{r'_1,r'_2,\cdots,r'_p\}$, the occurrence of a decomposition will be tested by checking if $\mathscr R^* \subseteq span B_1 \cup span B_2$. A positive test result means that $B_1$ and $B_2$ generate all the reaction vectors of the system, i.e., a binary decomposition $\{\mathscr N_1, \mathscr N_2\}$ is obtained where $\mathscr N_1$ and $\mathscr N_2$ have the reactions which corresponds to the vectors in $spanB_1$ and $spanB_2$, respectively. $B_1 \cap B_2 = \emptyset$ and $B_1 \cup B_2=B$ ensure that the decomposition is independent (STEPS 6-7).
    \item 	If in case the initial $B_1$ and $B_2$ will not be able to generate the desired decomposition, some vectors in $B_2$ will be transferred to $B_1$ and the process of checking for a possible decomposition will be repeated. This method of transferring vectors from $B_2$ to $B_1$ aims to preserve the disjointness of $B_1$ and $B_2$ while maintaining that their union is $B$ (STEP 8).
    \item Once a decomposition has been established, the generated subsystem $(\mathscr N_1,K_1)$ undergoes the building block tests. If the testing does not yield a positive result, the process will be repeated (STEPS 9-12).
\end{enumerate}

\subsection{The ACR building block algorithm}

We now present the ACR building block algorithm for PLK systems of SF-type.

\begin{algorithm} [ht]
\begin{algorithmic}

    \caption{ACR building block algorithm} 
    \STATE \textbf{INPUT}
        \STATE  Network: $\mathscr N$ with positive equilibrium 
        \STATE  Kinetic order matrix: $F$
        \STATE {}

 \STATE \textbf{STEP 1} Determine the following information about $\mathscr N$ 
             \STATE  Set of reaction vectors: $\mathscr R^* = \{y'-y \in \mathbb R^{\mathscr S}|y\rightarrow y' \in \mathscr R\}$
             \STATE  Rank of $\mathscr N$ 
    \STATE {}

    \STATE \textbf{STEP 2} Find all SF-pairs         
            \IF {there is none} 
                \STATE exit the algorithm 
            \ELSE
                \STATE store these SF-pairs as $\{\{R_u, R_v\}\}$ and proceed to STEP 3 
            \ENDIF 
    \STATE {}

    \STATE \textbf{STEP 3} Choose SF-pair $\{R_u, R_v\}$ 
         \STATE {}

    \STATE \textbf{STEP 4} Get the vectors $r_u$ and $r_v$ corresponding to $R_u$ and $R_v$
         \STATE {}

    \STATE \textbf{STEP 5} Extend $\{r_u,r_v\}$ to a basis $B=\{r_u,r_v,r'_1,r'_2,\cdots,r'_p\}$ of $S$
         \STATE {}   
        
    \STATE \textbf{STEP 6} Set $B_1=\{r_u,r_v\}$ and $B_2=\{r'_1,r'_2,\cdots,r'_p\}$
         \STATE {}

    \STATE \textbf{STEP 7} Test if $\mathscr R^* \subseteq spanB_1 \cup spanB_2$
            \IF {$\mathscr R^* \subseteq spanB_1 \cup spanB_2$} 
                \STATE proceed to STEP 9 
            \ELSE
                \STATE proceed to STEP 8 
            \ENDIF 
    \STATE {}

\end{algorithmic}
\end{algorithm}

\begin{algorithm}[!htbp]
\begin{algorithmic}

    \STATE \textbf{STEP 8} Choose a $P\in \wp(\{r'_1,r'_2,\cdots,r'_p\})-\{\emptyset,\{r'_1,r'_2,\cdots,r'_p\}\}$ where\\ $\wp(\{r'_1,r'_2,\cdots,r'_p\})$ is the set of all subsets of $\{r'_1,r'_2,\cdots,r'_p\}$\\ 
             \IF {$P$ has not been chosen before } 
                \STATE set $B_1=\{r_u,r_v\}\cup P$ and $B_2=\{r'_1,r'_2,\cdots,r'_p\}-P$ and repeat STEP 7
            \ELSE
                \IF {there is another $P\in \wp(\{r'_1,r'_2,\cdots,r'_p\})-\{\emptyset,\{r'_1,r'_2,\cdots,r'_p\}\}$ to choose} 
                    \STATE repeat STEP 8
                \ELSE 
                    \IF{there is another SF-pair that can be chosen}
                        \STATE repeat STEP 3
                    \ELSE
                        \STATE exit the algorithm 
                    \ENDIF
                \ENDIF
            \ENDIF 
    \STATE {} 
  
    \STATE \textbf{STEP 9} Get the subnetwork $\mathscr N_1$ generated by the reactions whose vectors are contained in $spanB_1$
    \STATE {}

    \STATE \textbf{STEP 10} Determine the deficiency $\delta_1$ of $\mathscr N_1$\\ 
            \IF {$\delta_1\leq 1$} 
                \STATE proceed to STEP 11 
            \ELSE
                \STATE repeat STEP 8
            \ENDIF 
    \STATE {}

    \STATE \textbf{STEP 11} Test if $(\mathscr N_1,K_1 )$ is a PL-RDK \\ 
        \IF {$(\mathscr N_1,K_1 )$ is a PL-RDK} 
                \STATE a building block is found and the system has ACR
            \ELSE
                \STATE proceed to STEP 12
            \ENDIF 
    \STATE {}

    \STATE \textbf{STEP 12} Test if $(\mathscr N_1,K_1 )$ is a minimally PL-NDK and $\delta_1=0$ \\ 
        
            \IF {$(\mathscr N_1,K_1 )$ is a minimally PL-NDK and $\delta_1=0$ } 
                \STATE a building block is found and the system has ACR
            \ELSE
                \STATE repeat STEP 8
            \ENDIF 
    \STATE {}

\end{algorithmic}
\end{algorithm}

\pagebreak

The following proposition proves that the decomposition produced in STEP 9 is an independent decomposition. 

\begin{proposition} \label{6} The decomposition $\{\mathscr N_1,\mathscr N_2\}$ induced by $B_1$ and $B_2$ where $\mathscr R^* \subseteq spanB_1 \cup spanB_2$ is an independent decomposition.
\end{proposition}

\begin{proof}
Clearly, $B_1$ and $B_2$ are nonempty. Now, note that $\mathscr N_1$ and $\mathscr N_2$ are generated by the reactions whose vectors are contained in $S_1=spanB_1$ and $S_2=spanB_2$, respectively. Let $R$ be any reaction in $\mathscr N$. Since $\mathscr R^* \subseteq spanB_1 \cup spanB_2$, then the reaction vector $r$ of $R$ is contained in either $S_1$ or $S_2$. Suppose $r\in S_1$. Then, $R$ is a reaction in $\mathscr N_1$. Now, if $r\in S_2$, then $r_u,r_v,r'_1,r'_2,\cdots,r'_p$ are going to be linearly dependent violating the fact that $B$ is a basis for $S$. This means that $R$ cannot be a reaction in $\mathscr N_2$ suggesting that $\{\mathscr N_1,\mathscr N_2\}$ is a decomposition. Finally, since $B_1\cap B_2=\emptyset$ and $B_1 \cup B_2=B$, then $S_1+S_2=S$ is a direct sum implying that $\{\mathscr N_1,\mathscr N_2\}$ is an independent decomposition. 
\end{proof}

\begin{example} \textbf{(from Running example 2)} We illustrate the algorithm by showing how it can produce the independent binary decomposition of the system in Running example 2 and subject the SF-type subnetwork to the building block tests. \\

\noindent STEP 1: The set of all reaction vectors of the system is $\mathscr R^* = \{r_1,r_2,\cdots, r_{17}\}$ where  \\
$r_1 = \begin{bmatrix} 1\\-1\\0\\0\\0\\0\\0\\0\\0\\0\\0\\0\\0\\ \end{bmatrix},
r_3 = \begin{bmatrix} 0\\0\\-1\\1\\0\\0\\0\\0\\0\\0\\0\\0\\0\\ \end{bmatrix},
r_5 = \begin{bmatrix} 0\\-1\\-1\\0\\1\\0\\0\\0\\0\\0\\0\\0\\0\\ \end{bmatrix},
r_7 = \begin{bmatrix} 1\\0\\0\\1\\-1\\0\\0\\0\\0\\0\\0\\0\\0\\ \end{bmatrix},
r_9 = \begin{bmatrix} 0\\0\\0\\0\\-1\\0\\0\\-1\\0\\0\\0\\1\\0\\ \end{bmatrix},\\
r_{11} = \begin{bmatrix} 0\\0\\0\\0\\1\\0\\0\\0\\1\\0\\0\\-1\\0\\ \end{bmatrix}, 
r_{12} = \begin{bmatrix} 0\\0\\0\\0\\0\\-1\\0\\0\\0\\0\\0\\-1\\1\\ \end{bmatrix},  
r_{13} = \begin{bmatrix} 0\\0\\0\\0\\0\\0\\1\\0\\0\\0\\0\\1\\-1\\ \end{bmatrix},
r_{14} = \begin{bmatrix} 0\\0\\0\\0\\0\\0\\-2\\0\\0\\1\\0\\0\\0\\ \end{bmatrix},
r_{15} = \begin{bmatrix} 0\\0\\0\\0\\0\\2\\0\\0\\0\\-1\\0\\0\\0\\ \end{bmatrix},\\
r_{16} = \begin{bmatrix} 0\\0\\0\\0\\0\\0\\-1\\0\\-1\\0\\1\\0\\0\\ \end{bmatrix},
r_{17} = \begin{bmatrix} 0\\0\\0\\0\\0\\1\\0\\1\\0\\0\\-1\\0\\0\\ \end{bmatrix},
r_2 = -r_1, r_4 = -r_3, r_6 = -r_5, r_8 = -r_7, r_{10} = -r_9.$\\

\noindent The system is of rank 9.\\

\noindent STEPS 2-6: The system has seven SF-pairs namely, $\{R_2,R_8\}, \{R_1,R_5\},$ $ \{R_3,R_5\}, \{R_4,R_8\},$ $\{R_{10},R_{12}\}, \{R_6,R_9\},$ and $\{R_7,R_9\}$. We deliberately choose $R_1$ and $R_5$. These reactions correspond to $r_1$ and $r_5$. We extend $\{r_1, r_5\}$ to a basis $B = \{r_1, r_5, r_3, r_9, r_{11}, r_{12}, r_{13}, r_{14}, r_{16}\}$ of the stoichiometric subspace $S$ of the system. Now, we set $B_1 = \{r_1, r_5\}$ and $B_2 = \{r_3, r_9, r_{11}, r_{12}, r_{13}, r_{14}, r_{16}\}$. It can be verified that $r_{7} \notin spanB_1\cup spanB_2$, i.e. $\mathscr R^* \nsubseteq spanB_1 \cup spanB_2$. Hence, we go to STEP 8.\\

\noindent STEP 8: Note that $\wp(\{r_3, r_9, r_{11}, r_{12}, r_{13}, r_{14}, r_{16}\})$ has $2^8=256$ elements. So, we have $254$ other iterations of $B_1$ and $B_2$. We can see in the following table a few of these iterations. From the table, notice that we get a decomposition when  $B_1 = \{r_1,r_5,r_3\}$ and $B_2 = \{r_9, r_{11}, r_{12}, r_{13}, r_{14}, r_{16}\}$. We proceed to STEP 9.

\begin{table}[H]
\begin{tabular}{|c|c|c|}
\hline
$B_1$ & $B_2$ & Element(s) not contained \\ 
 &  & in $spanB_1 \cup spanB_2$ \\ \hline
$\{r_1,r_5,r_3\}$           & $\{r_9, r_{11}, r_{12}, r_{13}, r_{14}, r_{16}\}$ & (none) \\ \hline
$\{r_1,r_5,r_9\}$           & $\{r_3, r_{11}, r_{12}, r_{13}, r_{14}, r_{16}\}$ & $r_7,r_8,r_{17}$ \\ \hline
$\{r_1,r_5,r_{11}\}$        & $\{r_3, r_9, r_{12}, r_{13}, r_{14}, r_{16}\}$ & $r_7,r_8,r_{17}$ \\ \hline
$\{r_1,r_5,r_{12}\}$        & $\{r_3, r_9, r_{11}, r_{13}, r_{14}, r_{16}\}$ & $r_7,r_8,r_{15}, r_{17}$ \\ \hline
$\{r_1,r_5,r_{13}\}$        & $\{r_3, r_9, r_{11}, r_{12}, r_{14}, r_{16}\}$ & $r_7,r_8,r_{15}, r_{17}$ \\ \hline
$\{r_1,r_5,r_{14}\}$        & $\{r_3, r_9, r_{11}, r_{12}, r_{13}, r_{16}\}$ & $r_7,r_8,r_{15}, r_{17}$ \\ \hline
$\{r_1,r_5,r_{16}\}$        & $\{r_3, r_9, r_{11}, r_{12}, r_{13}, r_{14}\}$ & $r_7,r_8,r_{17}$ \\ \hline
$\{r_1,r_5,r_3,r_9\}$       & $\{r_{11}, r_{12}, r_{13}, r_{14}, r_{16}\}$ & $r_{17}$ \\ \hline
$\{r_1,r_5,r_3,r_{11}\}$    & $\{r_9, r_{12}, r_{13}, r_{14}, r_{16}\}$ & $r_{17}$ \\ \hline
$\{r_1,r_5,r_3,r_{12}\}$    & $\{r_9, r_{11}, r_{13}, r_{14}, r_{16}\}$ & $r_{15}, r_{17}$ \\ \hline
$\{r_1,r_5,r_3,r_{13}\}$    & $\{r_9, r_{11}, r_{12}, r_{14}, r_{16}\}$ & $r_{15}, r_{17}$ \\ \hline
$\{r_1,r_5,r_3,r_{14}\}$    & $\{r_9, r_{11}, r_{12}, r_{13}, r_{16}\}$ & $r_{15}$ \\ \hline
$\{r_1,r_5,r_3,r_{16}\}$    & $\{r_9, r_{11}, r_{12}, r_{13}, r_{14}\}$ & $r_{17}$ \\ \hline
\end{tabular}
\centering
\caption{Some iterations of $B_1$ and $B_2$}
\label{table1}
\centering
\end{table}

\noindent STEPS 9-11: From STEP 8, we get the subnetwork $\mathscr N_1$ induced by the reactions whose vectors are contained in $spanB_1$ (see (\ref{eqn2})). As mentioned, $\mathscr N_1$ has deficiency $1$ and the subsystem $(\mathscr N_1, K_1)$ is a PL-RDK, i.e., a building block is found and the system has ACR in R2.
\end{example}

\section{The BCR building block algorithm: essential modifications}

By modifying some steps of the ACR building block algorithm, we get the BCR building block algorithm. The following are the alterations we have made in the procedure of ACR building block algorithm. We first remove the step that determines the set of reaction vectors and the rank of the network. Replace this step with the step that determines the incidence matrix $I_a$ and its rank. We take advantage of the column space ColS$(I_a)$ of $I_a$ here since a basis of ColS$(I_a)$ provides a basis for the image of $I_a$. Then, we replace items 4 and 5, with the following.

\begin{itemize} 
    \item From an SF-pair $\{R_u, R_v\}$, corresponding column vectors $c_u$ and $c_v$ of $I_a$ will be obtained. Then, $\{c_u,c_v\}$ will be extended to a basis $B = \{c_u,c_v,c'_1,c'_2,\cdots, c'_q\}$ of  ColS$(I_a)$ (STEPS 3-5).
    \item 	Let $B_1=\{c_u,c_v\}$ and $B_2=\{c'_1,c'_2,\cdots,c'_q\}$ and check if $C \subseteq span B_1 \cup span B_2$, where $C$ is the set of all column vectors of $I_a$. This testing aims to check if $spanB_1$ and $spanB_2$ generate all the column vectors of $I_a$. A binary decomposition $\{\mathscr N_1, \mathscr N_2\}$ is obtained when the test result is positive. Note that  $\mathscr N_1$ and $\mathscr N_2$  consist of the reactions which are associated to the column vectors in $spanB_1$ and  $spanB_2$, respectively. Finally, note that $B_1 \cap B_2 = \emptyset$ and $B_1 \cup B_2=B$ ensure that the decomposition is incidence independent (STEPS 6-7).
\end{itemize}

Similar to Proposition \ref{6}, the following result tells us that STEP 9 of BCR building block algorithm produced an incidence independent decomposition. 

\begin{proposition} \label{7} The decomposition $\{\mathscr N_1,\mathscr N_2\}$ induced by $B_1$ and $B_2$ where $C \subseteq spanB_1 \cup spanB_2$ is an incidence independent decomposition.
\end{proposition}

We now give the BCR building block algorithm for PLK systems of SF-type. 

\begin{algorithm} [H]
\begin{algorithmic}

    \caption{BCR building block algorithm}
    
    \STATE \textbf{INPUT}
        \STATE Network: $\mathscr N$ with complex balanced equilibrium 
        \STATE Kinetic order matrix: $F$
        \STATE {}
 
    \STATE \textbf{STEP 1} Determine the following information about $\mathscr N$ 
         \STATE Incidence matrix $I_a$
         \STATE Rank of $I_a$
         \STATE {}

\end{algorithmic}
\end{algorithm}

\begin{algorithm} [!htbp]
\begin{algorithmic}	
	
	\STATE \textbf{STEP 2} Find all SF-pairs 
	\IF {there is none} 
	\STATE exit the algorithm 
	\ELSE
	\STATE store these SF-pairs as $\{\{R_u, R_v\}\}$ and proceed to STEP 3 
	\ENDIF 
	\STATE {}

	\STATE \textbf{STEP 3} Choose SF-pair $\{R_u, R_v\}$ 
		\STATE {}

	\STATE \textbf{STEP 4} Get the column vectors $c_u$ and $c_v$ of $I_a$ corresponding to $R_u$ and $R_v$
		\STATE {}

	\STATE \textbf{STEP 5} Extend $\{c_u,c_v\}$ to a basis $B=\{c_u,c_v,c'_1,c'_2,\cdots,c'_q\}$ of ColS$I_a$
		\STATE {}   

	\STATE \textbf{STEP 6} Set $B_1=\{c_u,c_v\}$ and $B_2=\{c'_1,c'_2,\cdots,c'_q\}$.
		\STATE {}

\STATE \textbf{STEP 7} Test if $C \subseteq spanB_1 \cup spanB_2$.
\IF {$C \subseteq spanB_1 \cup spanB_2$} 
\STATE proceed to STEP 9 
\ELSE
\STATE proceed to STEP 8 
\ENDIF 
\STATE {}   

\STATE \textbf{STEP 8} Choose a $P\in \wp(\{c'_1,c'_2,\cdots,c'_p\})-\{\emptyset,\{c'_1,c'_2,\cdots,c'_q\}\}$ where\\ $\wp(\{c'_1,c'_2,\cdots,c'_p\})$ is the set of all subsets of $\{c'_1,c'_2,\cdots,c'_p\}$\\ 
\IF {$P$ has not been chosen before } 
\STATE set $B_1=\{c_u,c_v\}\cup P$ and $B_2=\{c'_1,c'_2,\cdots,c'_p\}-P$ and repeat STEP 7
\ELSE
\IF {there is another $P\in \wp(\{c'_1,c'_2,\cdots,c'_p\})-\{\emptyset,\{c'_1,c'_2,\cdots,c'_p\}\}$ to choose} 
\STATE repeat STEP 8
\ELSE 
\IF{there is another SF-pair that can be chosen}
\STATE repeat STEP 3
\ELSE
\STATE exit the algorithm 
\ENDIF
\ENDIF
\ENDIF 
\STATE {} 

\end{algorithmic}
\end{algorithm}    
    
\begin{algorithm}[h]
\begin{algorithmic}

    \STATE \textbf{STEP 9} Get the subnetwork $\mathscr N_1$ generated by the reactions which correspond to the column vectors of $I_a$ contained in $spanB_1$
         \STATE {}

    \STATE \textbf{STEP 10} Determine the deficiency $\delta_1$ of $\mathscr N_1$\\ 
            \IF {$\delta_1\leq 1$} 
                \STATE proceed to STEP 11 
            \ELSE
                \STATE repeat STEP 8
            \ENDIF 
        \STATE {}

        \STATE \textbf{STEP 11} Test if $(\mathscr N_1,K_1 )$ is a PL-RDK\\ 
        	\IF {$(\mathscr N_1,K_1 )$ is a PL-RDK} 
        \STATE a building block is found and the system has BCR
        	\ELSE
        \STATE proceed to STEP 12
        	\ENDIF 
        \STATE {}
        
        \STATE \textbf{STEP 12} Test if $(\mathscr N_1,K_1 )$ is a minimally PL-NDK and $\delta_1=0$\\ 
        	\IF {$(\mathscr N_1,K_1 )$ is a minimally PL-NDK and $\delta_1=0$ }
        \STATE a building block is found and the system has BCR
        	\ELSE
        \STATE repeat STEP 8
        	\ENDIF 
        \STATE {}
        
\end{algorithmic}
\end{algorithm}

\pagebreak

\begin{example} \textbf{(from Running example 1)} Consider the system given in Running example 1. See (\ref{graph100}) for the CRN and (\ref{graph300}) for the kinetic order matrix. It is a deficiency zero PL-NDK system which contains two NDK nodes: $B$ and $B+F$. It has a positive equilibrium according to Theorem 4 of \cite{fort2}. By Theorem \ref{1.5}, the equilibrium is also complex balanced. Using the BCR building block algorithm, we show that the system has BCR in $F$.\\

\noindent STEP 1: The incidence matrix, given in (\ref{eqn9}), is of rank 5.\\

\noindent STEPS 2-7: The system has six SF-pairs, namely, $\{R_2, R_4\}$, $\{R_2, R_6\}$, $\{R_4, R_5\}$, $\{R_4, R_{10}\}$, $\{R_5, R_8\}$, $\{R_8, R_{10}\}$. Being one of the two SF-pairs that corresponds to $F$, we pick $\{R_8, R_{10}\}$. We consider $c_8$ and $c_{10}$, the column vectors of $I_a$ that correspond to $R_8$ and $R_{10}$. We get the basis $B=\{c_8, c_{10}, c_1, c_3, c_5\}$ of ColS$(I_a)$ after extending $\{c_8, c_{10}\}$. Let $B_1=\{c_8,c_{10}\}$ and $B_2=\{c_1,c_3,c_5\}$. It can be noticed that $c_7,c_8,c_9,c _{10}\in spanB_1$ and $c_1, c_2, c_3, c_4, c_5, c_6 \in spanB_2$ suggesting that $spanB_1\cup spanB_2$ generates all the column vectors of $I_a$. Then, we proceed to STEP 9.\\

\noindent STEPS 9-12: Consider the subnetwork $\mathscr N_1$ (see \ref{eqn10}) generated by the reactions associated to the vectors in $spanB_1$. The deficiency of $\mathscr N_1$ is $\delta_1 = 3-2-1=0\leq 1$. However, it is PL-NDK having $B+F$ as the sole NDK node. This NDK node has two CF-subsets, namely, $\{R_8\}$ and $\{R_{10}\}$ where each contains exactly one reaction. Hence, $\mathscr N_1$ is a minimally PL-NDK which means that we found a building block and so the given system has BCR in $F$.
 
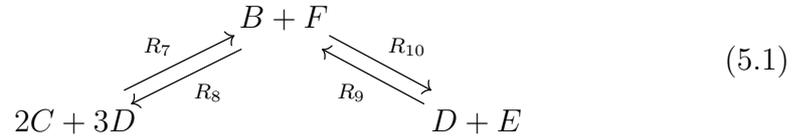
\begin{equation}
    \begin{tikzcd}
                                    & B+F \arrow[ld, "R_8", shift left] \arrow[rd, "R_{10}", shift left] &                                   \\
2C+3D \arrow[ru, "R_7", shift left] &                                                                    & D+E \arrow[lu, "R_9", shift left]
\end{tikzcd}
\label{eqn10}
\end{equation}
\end{example}

\vspace{.5cm}

\begin{example} The following PLK system was used as an example in \cite{mend}. 
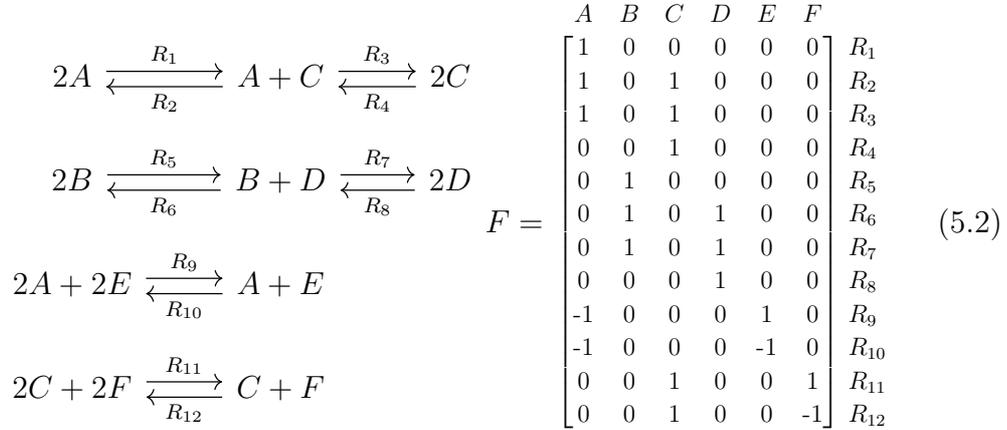
\begin{equation}
    \begin{tikzcd}
2A \arrow[r, "R_1", shift left]       & A+C \arrow[l, "R_2", shift left] \arrow[r, "R_3", shift left] & 2C \arrow[l, "R_4", shift left] \\
2B \arrow[r, "R_5", shift left]       & B+D \arrow[l, "R_6", shift left] \arrow[r, "R_7", shift left] & 2D \arrow[l, "R_8", shift left] \\
2A+2E \arrow[r, "R_9", shift left]    & A+E \arrow[l, "R_{10}", shift left]                           &                                 \\
2C+2F \arrow[r, "R_{11}", shift left] & C+F \arrow[l, "R_{12}", shift left]                           &                                
    \end{tikzcd}
F =
\scalebox{.79}{
  \begin{blockarray}{*{6}{c} l}
    \begin{block}{*{6}{>{$\footnotesize}c<{$}} l}
      A & B & C & D & E & F \\
    \end{block}
    \begin{block}{[*{6}{c}]>{$\footnotesize}l<{$}}
      1 & 0 & 0 & 0 & 0 & 0 & R_1 \\
      1 & 0 & 1 & 0 & 0 & 0 & R_2 \\
      1 & 0 & 1 & 0 & 0 & 0 & R_3 \\
      0 & 0 & 1 & 0 & 0 & 0 & R_4 \\
      0 & 1 & 0 & 0 & 0 & 0 & R_5 \\
      0 & 1 & 0 & 1 & 0 & 0 & R_6 \\
      0 & 1 & 0 & 1 & 0 & 0 & R_7 \\
      0 & 0 & 0 & 1 & 0 & 0 & R_8 \\
      -1 & 0 & 0 & 0 & 1 & 0 & R_9 \\
      -1 & 0 & 0 & 0 & -1 & 0 & R_{10} \\
      0 & 0 & 1 & 0 & 0 & 1 & R_{11} \\
      0 & 0 & 1 & 0 & 0 & -1 & R_{12} \\
    \end{block}
  \end{blockarray}}
  \label{eqn6} 
\end{equation}

\noindent It has a complex balanced equilibrium and deficiency $\delta = 3$. Let us show using BCR building block algorithm that this system has a BCR in $C$. The system's incidence matrix $I_a$ of rank 6 is given by the following.
\begin{equation}
I_a = 
    \begin{blockarray}{*{12}{c} l}
        \begin{block}{*{12}{>{$\footnotesize}c<{$}} l}
            $R_1$ & $R_2$ & $R_3$ & $R_4$ & $R_5$ & $R_6$ & $R_7$  & $R_8$ & $R_9$ & $R_{10}$ & $R_{11}$ & $R_{12}$\\
        \end{block}
        \begin{block}{[*{12}{c}]>{$\footnotesize}l<{$}}
            -1& 1 & 0 & 0 & 0 & 0 & 0 & 0 & 0 & 0 & 0 & 0 & $2A$\\
            1 & -1& -1& 1 & 0 & 0 & 0 & 0 & 0 & 0 & 0 & 0 & $A+C$\\
            0 & 0 & 1 & -1& 0 & 0 & 0 & 0 & 0 & 0 & 0 & 0 & $2C$\\
            0 & 0 & 0 & 0 & -1& 1 & 0 & 0 & 0 & 0 & 0 & 0 & $2B$\\
            0 & 0 & 0 & 0 & 1 & -1& -1& 1 & 0 & 0 & 0 & 0 & $B+D$\\
            0 & 0 & 0 & 0 & 0 & 0 & 1 & -1& 0 & 0 & 0 & 0 & $2D$\\
            0 & 0 & 0 & 0 & 0 & 0 & 0 & 0 &-1 & 1 & 0 & 0 & $2A+2E$\\
            0 & 0 & 0 & 0 & 0 & 0 & 0 & 0 & 1 &-1 & 0 & 0 & $A+E$\\
            0 & 0 & 0 & 0 & 0 & 0 & 0 & 0 & 0 & 0 & -1& 1 & $2C+2F$\\
            0 & 0 & 0 & 0 & 0 & 0 & 0 & 0 & 0 & 0 & 1 & -1& $C+F$\\
        \end{block}
    \end{blockarray}
\label{eqn7}
\end{equation}

The system has the following SF-pairs: $\{R_1,R_2\}, \{R_1,R_3\}, \{R_2,R_4\}, \{R_3,R_4\}, \{R_5,R_6\},$ \\$\{R_5,R_7\}, \{R_6,R_8\}, \{R_7,R_8\}, \{R_9,R_{10}\}, \{R_{11},R_{12}\}$. We choose the pair $\{R_1,R_2\}$ since the kinetic order vectors of $R_1$ and $R_2$ differ only at $C$. The column vectors $c_1$ and $c_2$ of $I_a$ correspond to $R_1$ and $R_2$. Extending $\{c_1,c_2\}$ to a basis of ColS$(I_a)$, we get $B = \{c_1,c_3,c_5,c_7,c_9,c_{11}\}$. We let $B_1=\{c_1,c_2\}$ and $B_2 = \{c_3,c_5,c_7,c_9,c_{11}\}$. Observe that $spanB_1$ contains $c_1$ and $c_2$ while $spanB_2$ contains $c_3,c_4,c_5,\cdots,c_{12}$, i.e., $spanB_1\cup spanB_2$ contains all the column vectors of the matrix of $I_a$. Hence, we proceed by subjecting  the subnetwork $\mathscr N_1$ (see (\ref{eqn8})) generated by the reactions whose corresponding column vectors are contained in $spanB_1$ to the building block tests. 
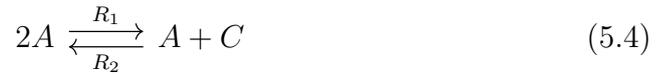
\begin{equation}
    \begin{tikzcd}
        2A \arrow[r, "R_1", shift left] & A+C \arrow[l, "R_2", shift left]
    \end{tikzcd}
\label{eqn8}
\end{equation}
\noindent $\mathscr N_1$ has deficiency $\delta_1 = 2-1-1 = 0$  and clearly, $(\mathscr N_1,K_1)$ is a PL-RDK. So, a building block is found, and we conclude that the system has a BCR in $C$.
\end{example}

\section{Summary and outlook}

We summarize the results we have obtained and give some recommendations for further research.

 We have developed algorithms that enable us to identify ACR and BCR on PLK systems. The algorithms are based on Propositions 8 and 9 of \cite{fort3}. Each of the algorithms works by searching for an SF-type subnetwork taken from a suitable decomposition. This subnetwork undergoes several tests to determine if it is a building block.

 While the algorithms can determine the existence of a concentration robustness in a PLK system given that the system has all the required conditions, these  procedures are not capable of determining if a PLK system does not have a concentration robustness.

 ACR building block algorithm is also capable of finding every binary independent decomposition (such that one subnetwork of each decomposition contains a fixed pair of reactions) of a given network. Hence, as implied by Proposition \ref{5} (Proposition 5 in \cite{fari}), this algorithm also enables us to detect if a network does not have a more refined independent decomposition. In other words, if a binary independent decomposition does not exist, then an independent refinement of such decomposition also does not exist. Similarly, BCR building block algorithm can also be used to determine if a network does not have a refined incidence independent decomposition. These observations imply that the algorithms can be utilized as tools in analyzing decompositions of a network. In this regard, further studies could be done to achieve computational techniques that generate all independent and incidence independent decompositions of a network. 

 Each of the algorithms here generates decompositions by considering subsets of a power set. Since power set easily gets large, a computationally less expensive strategy in finding those subsets will improve the efficiency of the algorithms. 

 The creation of computer programs applying the algorithms is a good next step for this study. If successful, such project will be beneficial for those who study concentration robustness in PLK systems, or other systems that can be represented as PLK systems.

\vspace{0.5cm}
\noindent \textbf{Acknowledgement}
L. L. Fontanil extends his gratitude to the Department of Science and Technology-Science Education Institute (DOST-SEI), Philippines for supporting him through the ASTHRDP Scholarship grant.

\baselineskip=0.25in

\appendix

\section*{Appendix}

We discuss here a system where the ACR building block algorithm is silent. Consider the PLK system for R. Schmitz's preindustrial carbon cycle model which appeared as an example in \cite{fort2}. The system has a positive equilibrium and is of deficiency 0.

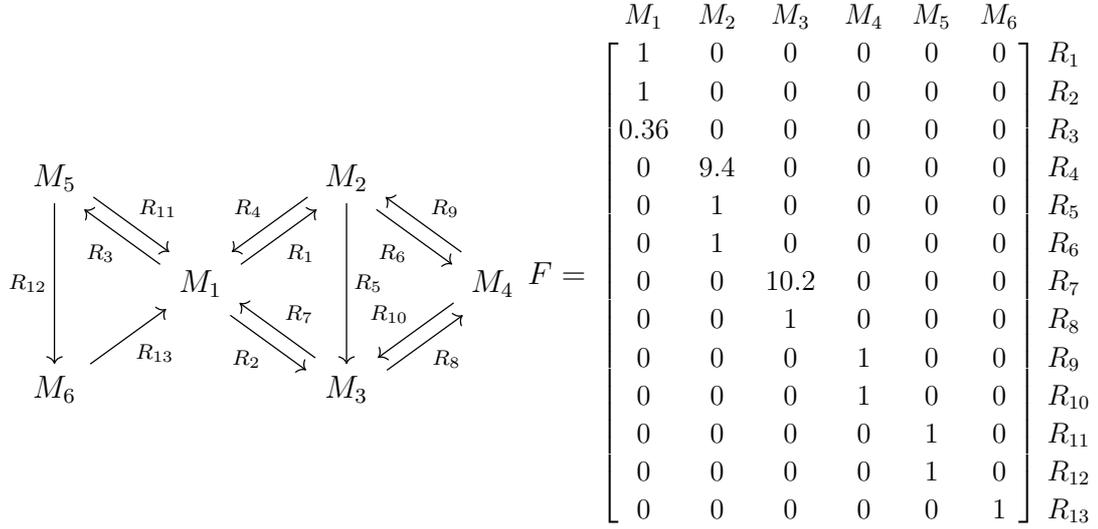
\begin{equation}
\nonumber
    \begin{tikzcd}
M_5 \arrow[dd, "R_{12}"'] \arrow[rd, "R_{11}", shift left] &                                                                                                   & M_2 \arrow[ld, "R_4"', shift right] \arrow[rd, "R_6"', shift right] \arrow[dd, "R_5"] &                                                                        \\
                                                           & M_1 \arrow[lu, "R_3", shift left] \arrow[ru, "R_1"', shift right] \arrow[rd, "R_2"', shift right] &                                                                                       & M_4 \arrow[ld, "R_{10}"', shift right] \arrow[lu, "R_9"', shift right] \\
M_6 \arrow[ru, "R_{13}"']                                  &                                                                                                   & M_3 \arrow[lu, "R_7"', shift right] \arrow[ru, "R_8"', shift right]                   &                                                                       
\end{tikzcd}
F =
\scalebox{.9}{
  \begin{blockarray}{*{6}{c} l}
    \begin{block}{*{6}{>{$\footnotesize}c<{$}} l}
      M_1 & M_2 & M_3 & M_4 & M_5 & M_6 \\
    \end{block}
    \begin{block}{[*{6}{c}]>{$\footnotesize}l<{$}}
      1 & 0 & 0 & 0 & 0 & 0 & R_1 \\
      1 & 0 & 0 & 0 & 0 & 0 & R_2 \\
      0.36&0& 0 & 0 & 0 & 0 & R_3 \\
      0 &9.4& 0 & 0 & 0 & 0 & R_4 \\
      0 & 1 & 0 & 0 & 0 & 0 & R_5 \\
      0 & 1 & 0 & 0 & 0 & 0 & R_6 \\
      0 & 0 &10.2&0 & 0 & 0 & R_7 \\
      0 & 0 & 1 & 0 & 0 & 0 & R_8 \\
      0 & 0 & 0 & 1 & 0 & 0 & R_9 \\
      0 & 0 & 0 & 1 & 0 & 0 & R_{10} \\
      0 & 0 & 0 & 0 & 1 & 0 & R_{11} \\
      0 & 0 & 0 & 0 & 1 & 0 & R_{12} \\
      0 & 0 & 0 & 0 & 0 & 1 & R_{13} \\
    \end{block}
  \end{blockarray}}
  \label{eqn4000} 
\end{equation}

\noindent STEP 1: The set of reaction vectors is given by $\mathscr R^*=\{r_1, r_2, \cdots, r_{13}\}$ where  \\
$r_1 = \begin{bmatrix} -1\\1\\0\\0\\0\\0\\ \end{bmatrix},
r_2 = \begin{bmatrix} -1\\0\\1\\0\\0\\0\\ \end{bmatrix},
r_3 = \begin{bmatrix} -1\\0\\0\\0\\1\\0\\ \end{bmatrix},
r_4 = -r_1, 
r_5 = \begin{bmatrix} 0\\-1\\1\\0\\0\\0\\ \end{bmatrix},
r_6 = \begin{bmatrix} 0\\-1\\0\\1\\0\\0\\ \end{bmatrix},\\
r_7 = -r_2,$
$r_8 = \begin{bmatrix} 0\\0\\-1\\1\\0\\0\\ \end{bmatrix},
r_9 = -r_6,
r_{10} = -r_8,
r_{11} = -r_3,
r_{12} = \begin{bmatrix} 0\\0\\0\\0\\-1\\1\\ \end{bmatrix},\\
r_{13} = \begin{bmatrix} 1\\0\\0\\0\\0\\-1\\ \end{bmatrix}$

\noindent STEPS 2-7: The system has five SF-pairs, but we only choose the pair $\{R_4,R_5\}$. Reactions $R_4$ and $R_5$ correspond to $r_4$ and $r_5$. We get a basis $B = \{r_4,r_5,r_3,r_8,r_{12}\}$ of the stoichiometric subspace after extending $\{r_4,r_5\}$. We let $B_1 = \{r_4, r_5\}$ and $B_2 = \{r_3,r_8,r_{12}\}$. Notice that $r_4,r_5,r_1,r_2,r_7 \in spanB_1$ while $r_3,r_8,r_{12},r_{11},r_{10},r_{13} \in spanB_2$. On the other hand, $r_6,r_9 \notin spanB_1 \cup spanB_2$, i.e., $\mathscr R^* \nsubseteq spanB_1 \cup spanB_2$. We proceed to STEP 8.\\

\noindent STEP 8: Before transferring some elements of  $B_2$ to $B_1$, we point out that $\wp(B_2)$ has $2^3=8$ elements. Excluding $\emptyset$ and $B_2$ gives us 6 ways to rewrite $B_1$ and $B_2$. We can see from the table below that only the case where $B_1 = \{r_4,r_5,r_8\}$ and $B_2 = \{r_3,r_{12}\}$ induces a partition of $\mathscr R^*$ resulting to a decomposition of the given network. Since every other iteration of $B_1$ and $B_2$ fails to decompose the given network, this step of the algorithm will ultimately pick the case where $B_1 = \{r_4,r_5,r_8\}$ and $B_2 = \{r_3,r_{12}\}$.

\begin{table}[H]
\begin{tabular}{|c|c|c|}
\hline
$B_1$ & $B_2$ & element(s) not contained \\ 
      &       & in $spanB_1 \cup spanB_2$ \\ \hline
$\{r_4,r_5,r_3\}$ & $\{r_8,r_{12}\}$ & $r_6,r_9,r_{13}$ \\ \hline
$\{r_4,r_5,r_8\}$ & $\{r_3,r_{12}\}$ & (none) \\ \hline
$\{r_4,r_5,r_{12}\}$ & $\{r_3,r_8\}$ & $r_6,r_9,r_{13}$ \\ \hline
$\{r_4,r_5,r_3,r_8\}$ & $\{r_{12}\}$ & $r_{13}$ \\ \hline
$\{r_4,r_5,r_3,r_{12}\}$  & $\{r_8\}$ & $r_6,r_9,r_{10},r_{13}$ \\ \hline
$\{r_4,r_5,r_8,r_{12}\}$ & $\{r_3\}$ & $r_{11},r_{13}$ \\ \hline
\end{tabular}
\centering
\caption{Some iterations of $B_1$ and $B_2$}
\label{table2}
\end{table}

\noindent STEPS 9-10: Given below is the subnetwork $\mathscr N_1$ generated by the reactions whose vectors are contained in $spanB_1 = span\{r_4,r_5,r_8\}$. It has deficiency $\delta_1 = 4-1-3 = 0 \leq 1$.

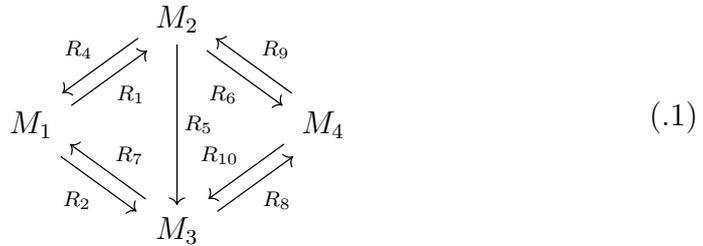
\begin{equation}
    \begin{tikzcd}
                                                                    & M_2 \arrow[ld, "R_4"', shift right] \arrow[rd, "R_6"', shift right] \arrow[dd, "R_5"] &                                                                        \\
M_1 \arrow[ru, "R_1"', shift right] \arrow[rd, "R_2"', shift right] &                                                                                       & M_4 \arrow[ld, "R_{10}"', shift right] \arrow[lu, "R_9"', shift right] \\
                                                                    & M_3 \arrow[lu, "R_7"', shift right] \arrow[ru, "R_8"', shift right]                   &                                                                       
\end{tikzcd}
\end{equation}

\noindent STEPS 11-12: Notice that in the subsystem $(\mathscr N_1, K_1)$, the branching reactions $R_4$ and $R_5$ have different kinetic order vectors. This means that the subsystem is NOT PL-RDK. Unfortunately, $M_2$ and $M_3$ are both NDK nodes in the subsystem implying that it is also not minimally PL-NDK. No building block is found and the algorithm cannot determine if the system has ACR using the SF-pair $\{R_4,R_5\}$.\\

We note that testing every other SF-pair yields similar result, i.e., no decomposition gives an SF-type subnetwork that passes the building block tests. In short, the algorithm cannot determine if the given system has ACR. Nevetheless, though the result is not affirmative, this example also illustrates the capability of the algorithm (and the BCR building block algorithm as well) to extract and analyze every binary decomposition of the network. This feature of the algorithm can be seen as an opportunity for those who are studying decomposition theory.
 
\end{document}